\documentclass[12pt,letterpaper,notitlepage]{article}
\usepackage[utf8]{inputenc}
\usepackage[english]{babel}
\usepackage[T1]{fontenc}
\usepackage{bm,bbm,mathbbol}
\usepackage{amsbsy,amsmath,amsthm, amssymb,amsfonts}
\usepackage{graphicx}
\usepackage{subfigure}
\usepackage{verbatim}
\usepackage{url}

\usepackage[hmargin=3.1cm,vmargin={2.2cm,1.7cm},includefoot]{geometry}

\newcommand{\N}{\mathbb{N}}
\newcommand{\C}{\mathbb{C}}
\newcommand{\R}{\mathbb{R}}

\newcommand{\half}{\mathbb{H}}

\newcommand{\strip}{{S_\pi}}

\newcommand{\E}{\mathbb{E}}

\newcommand{\oo}{o}
\newcommand{\de}{\mathrm{d}}

\renewcommand{\Im}{\mathrm{Im}}

\theoremstyle{plain}
\newtheorem{theorem}{Theorem}

\theoremstyle{definition}

\title{Stationarity of SLE}
\author{Antti Kemppainen\thanks{
Email: \protect\url{antti.h.kemppainen@helsinki.fi}.
Department of Mathematics and Statistics, University of Helsinki.}}
\date{}

\begin{document}

\maketitle

\begin{abstract}
A new method to study a stopped hull of SLE$_\kappa(\rho)$ is presented. In this approach, the law of the conformal map associated to the hull
is invariant under a SLE induced flow. The full trace of a chordal SLE$_\kappa$ can be studied using this approach.
Some example calculations are presented.
\end{abstract}

\section{Introduction}

Schramm--Loewner evolution (SLE) was introduced by Oded Schramm \cite{schramm-2000-}. SLEs are random curves in
the plane. There are many variants of SLE, but the local properties of the random curve are determined by a single parameter
$\kappa \geq 0$. SLEs are characterized by conformal invariance and the domain Markov property. The scaling limits of
two-dimensional statistical physics models at criticality are believed to be conformally invariant. For this reason the scaling
limit of a curve emerging from such a model has to be SLE$_\kappa$ for some $\kappa \geq 0$. The parameter $\kappa$ describes the
universality class of the model.

A chordal SLE is a random curve in a simply connected domain connecting two boundary points. 
In the section~\ref{sec: sle}, we will define the chordal SLE in more detail.
The chordal SLE is \emph{stationary} in the sense that given the process up to a time $t$ the law of $K_{t + s}$
is such that $g_t(K_{t + s} \setminus K_t) - X_t$ and $K_s$ have the same law, where 
$(g_t)_{t \geq 0}$ is the collection of conformal mappings satisfying the Loewner equation, 
$(K_t)_{t \geq 0}$ denotes the corresponding collection of subsets of the upper half-plane $\half$,
and $(X_t)_{t \geq 0}$ is the driving process.

SLE$_\kappa( \rho )$-processes, $\kappa \geq 0$ and $\rho \in \R$, are generalizations of the chordal SLE$_\kappa$.
When $\rho=0$ this reduces to the chordal case: SLE$_\kappa( 0)$ is the chordal SLE$_\kappa$. 
The definition of SLE$_\kappa( \rho )$ requires two marked points.
If $X_t$ is the driving process of a SLE$_\kappa( \rho )$ and the other marked point is $Y_t$, 
then for a range of the parameter values the hitting time
$\tau = \inf \{ t \geq 0 : |Y_s - X_s| \to 0 \textrm{ as } s \nearrow t \}$
is almost surely finite. The stopped hull $K_\tau$ is a interesting object in many ways. For example,  SLE$_\kappa(\kappa -6)$
is a coordinate transformation of the chordal SLE$_\kappa$ and hence $K_\tau$ describes the full SLE$_\kappa$ trace
seen from a fixed point in the real axis.

The novel result of this paper is a formulation of the \emph{stationarity} of SLE$_\kappa( \rho )$ in Theorem~\ref{thm: stationarity} 
so that $K_\tau$ is invariant under a flow induced the SLE. 
In this approach, the SLE is run for a time $t>0$, then this beginning is erased, and scaling and translation are used to map
the beginning and end points $X_t$ and $Y_t$ back to the initial values $X_0$ and $Y_0$. 
By the property stated in Theorem~\ref{thm: stationarity}, 
$(g_t( K_\tau \setminus K_t) - \beta_t) /\alpha_t$ has the same law as $K_\tau$, where $\alpha_t$ and $\beta_t$ are the
appropriate scaling and translation factors.

Theorem~\ref{thm: stationarity} enables us to calculate quantities related to $K_\tau$
such as the moments $\E[ \prod_{j=1}^n a_{k_j} ]$ of the coefficient of the expansion $G(z) = g_\tau(z) = z + \sum_j a_j z^{-j}$.
The driving function and the coefficients of the Loewner map can be viewed as the ``state of SLE''
and they form the SLE data.
The stationarity gives a new way to calculate the distribution functions or the expected values of the SLE data.
This is related to the approach in \cite{kytola-kemppainen-2006-},
although the work of this paper was done before that paper.

In the sections~\ref{ssec: reversibility} and \ref{ssec: general moments}, 
an approach for the reversibility of the chordal SLE  is proposed, 
and for $\rho = \kappa -6$, the general form of $\E[ \prod_{j=1}^n a_{k_j} ]$ as a function of $\kappa$ is derived using the reversibility.
The reversibility was recently proven to hold for chordal SLE$_\kappa$, $\kappa \in [0,4]$ by Dapeng Zhan \cite{zhan-2008-}.
It is a property of SLE
that states that if the roles of the beginning and end points are changed, then the law of the
random curve remains the same.

In the section~\ref{ssec: a1 a2 moments}, moments of the form $\E[ a_1^n ]$ and $\E[ a_1^n a_2^m ]$ are calculated.
In the section~\ref{ssec: a1 density}, the method is used to derive the distribution of $a_1$.

\section{SLE and Schramm's principle} \label{sec: sle}

\subsection{Chordal SLE}

One natural choice for a simply connected domain in the complex plane having two marked boundary points is
the upper half-plane $\half = \{ z \in \C \,:\, \Im(z) > 0 \}$. The marked points are $0$ and $\infty$.
The triplet $(\half,0,\infty)$ is preserved by the family of mappings $z \mapsto \lambda z, \lambda>0$.
The Schwarz lemma shows that these are the only conformal mappings with this property.

A subset $K \subset \half$ is a \emph{hull} if $K = \half \cap \overline{K}$, $K$ is bounded and $\half \setminus K$
is simply connected. If $\gamma:[0,T] \to \C$ is a simple curve such that $\gamma(0) \in \R$ and $\gamma(0,T] \subset \half$,
then $K_t = \gamma(0,t]$ is a hull for each $t \in [0,T]$. In this case the family $(K_t)_{t \in [0,T]}$ is growing in the sense
that $K_t \subsetneq K_s$ when $0 \leq t < s \leq T$.

Let $(K_t)_{t \geq 0}$ be a growing family of hulls and $g_t$ be the conformal mapping from $\half \setminus K_t$ 
onto $\half$  that is normalized by $g_t(z) = z + \oo(1)$ as $z \to \infty$. This normalization makes $g_t$ unique.
If $K_0 = \emptyset$ and $(K_t)_{t \geq 0}$ grows continuously in a quite natural sense,
we can reparameterize $K_t$ so that $g_t (z) = z + 2t/z + \ldots$ at infinity. 

If $(K_t)_{t \geq 0}$ grows locally in the sense of Theorem~2.6 of \cite{lawler-schramm-werner-2001-}
then the family of mappings $(g_t)_{t \geq 0}$ satisfies the upper half-plane Loewner equation
\begin{equation} \label{eq: loewner}
\partial_t g_t(z) = \frac{2}{g_t(z) -X_t} 
\end{equation}
where $X_t \in \R$ is called the driving function (process) of $K_t$. In fact $X_t$ is the image
of the point where $K_t$ is growing under the mapping $g_t$, that is $X_t =\cap_{s > t} \overline{ g_t(K_s \setminus K_t) }$.
Note that the family of hulls given by a simple curve is growing locally.

Consider now a collection of probability measures $(\mu_{\Omega,a,b})$ such that
$\mu_{\Omega,a,b}$ is the law of a random curve in $\overline{\Omega}$ connecting two boundary points
$a$ and $b$ of a simply connected domain $\Omega$. Choose some consistent parameterization for such curves
so that they are parametrized by $t \in [0,\infty)$.
Now we use \emph{Schramm's principle} 
(which appeared in the seminal paper \cite{schramm-2000-} by Schramm, see e.g. the discussion about LERW in the
introduction of that paper. It is formulated in the following way in \cite{smirnov-2006-}.)
and we demand that $(\mu_{\Omega,a,b})$ satisfies the following two requirements:
\begin{description}
\item[(CI) Conformal invariance:] For any triplet $(\Omega,a,b)$ and any conformal mapping $\phi: \Omega \to \C$,
it holds that $\phi \mu_{\Omega,a,b} = \mu_{\phi (\Omega), \phi(a), \phi(b)}$.
\item[(DMP) Domain Markov property:] 
Suppose we are given $\gamma [0,t]$, $t>0$.
The conditional law of $\gamma(t+s)$ given $\gamma[0,t]$ 
is the same as the law of $\gamma(s)$ in the slit domain $(\Omega \setminus \gamma[0,t],\gamma(t),b)$. That is
\begin{equation*}
\mu_{\Omega,a,b} ( \;\cdot\; | \,\gamma[0,t] ) = \mu_{\Omega \setminus \gamma[0,t],\gamma(t),b} 
\end{equation*}
\end{description}
First of all CI tells that $\mu_{\Omega,a,b} = \phi \mu_{\half,0,\infty}$, where $\phi$ is a conformal mapping from
the triplet $(\half,0,\infty)$ to the triplet $(\Omega,a,b)$. Note that $\phi$ is not unique: any
$\phi(\lambda \,\cdot\,), \lambda>0$ would also do. So for each $(\Omega,a,b)$ choose some $\Phi=\phi$.

Now we can restrict to the standard triplet $(\half,0,\infty)$. Let $H_t$ be the unbounded component
of $\half \setminus \gamma[0,t]$, $K_t$ the complement of $H_t$ in $\half$ and $g_t$ the mapping associated with $K_t$.
The combination of CI and DMP shows that
the curve $\tilde \gamma : s \mapsto g_t \big(\gamma(t + s)\big)  - X_t$ is independent of $\gamma[0,t]$ and is identically
distributed to $\gamma$. This leads to the fact that $X_t$ has independent and stationary increments. 
Since $K_t$, defined by a curve, is growing locally, it has a continuous driving process. 
All the continuous processes with independent and stationary increments are of the form
\begin{displaymath}
X_t = \sqrt{\kappa} B_t + \theta t,
\end{displaymath}
with some constants $\kappa\geq 0$ and $\theta \in \R$. Here $B_t$ is a standard one-dimensional Brownian motion. Let $\phi_\lambda : z 
\mapsto \lambda z$. CI with $\phi=\phi_\lambda$ implies that $X_t$ and $\lambda X_{t/\lambda^2}$ have the same law. This
shows that $\theta=0$ and furthermore that the law of the random curve in $(\Omega,a,b)$ doesn't depend
on the choice of $\Phi$.

Chordal SLE$_\kappa$ is the law of $K_t$ with the driving process $X_t = \sqrt{\kappa} B_t$.
It turns out that $K_t$ is generated by a curve in the sense that there is a curve $\gamma$ so that
$\half \setminus K_t$ is the unbounded component
of $\half \setminus \gamma[0,t]$, see \cite{rohde-schramm-2005-}. Such $\gamma$ is called the \emph{trace}.
For $\kappa \in (0,4]$ it is a simple curve.

\subsection{Strip SLE and the upper-half plane SLE$_\kappa (\rho)$}

It is possible to repeat Schramm's principle for three marked boundary points.
A natural domain for three marked points is the infinite strip $\strip = \{z \in \C \,:\, 0<\Im (z) < \pi\}$. The marked points are now
$0, -\infty$ and $+\infty$.

We can continue in the same way as in the case of the upper half-plane. For a family of hulls $(K_t)_{t \geq 0}$
on the strip $\strip$, let $g^\strip_t$ be a conformal mapping from $\strip \setminus K_t$ onto $\strip$
normalized by $g^\strip_t(z) = z \pm const. + \oo(1)$ as $z \to \pm \infty$. We can reparameterize such that
$g^\strip_t(z) = z \pm t + \oo(1)$ as $z \to \pm \infty$. 
The strip Loewner equation is
\begin{equation}
\partial_t g^\strip_t(z) = \coth \left( \frac{g^\strip_t(z) - X_t}{2} \right) .
\end{equation}
We can formulate the conformal invariance and the domain Markov property for three marked points by adding a third point $c$ which
behaves the same way as $b$. 
As in the two point case we can show that the collection of probability measures $(\mu_{\Omega,a,b,c})$ has properties
CI and DMP if and only if the driving process of the random curve of $\mu_{\strip,0,\infty,-\infty}$ is 
of the form
\begin{equation*}
X_t = \sqrt{\kappa} B_t + \theta t. 
\end{equation*}
Now we don't have any conformal mappings other than the identity map preserving
$(S_\pi,0,-\infty,+\infty)$. So in general,  $\theta$ doesn't need to vanish. 
Hence the strip SLEs are a family of probability measures parameterized by two real parameters. 
See also \cite{schramm-wilson-2005-}.


The infinite strip $\strip$ can be mapped to the upper half-plane by mappings of the form $\phi: z \mapsto \alpha e^{\pm z} +\beta$
where $\alpha, \beta \in \R$ and the sign of $\alpha$ is such that $i \pi / 2$ is mapped to the upper half-plane.
Choose $\alpha$ and $\beta$ so that the marked points are mapped in the following way:
$0$ to $x \in \R$  and one of $-\infty$ or $+\infty$ to $\infty$ and the other to $y \in \R$.
The strip SLE is mapped to a random curve of the upper half-plane by defining $\widehat{K}_t = \phi( K_t)$
which is a collection of hulls of $\half$ parametrized by the ``strip capacity''.
After a time change to the upper half-plane capacity,
the half-plane mappings $g_t$ related to these hulls satisfy the half-plane Loewner equation~\eqref{eq: loewner} with
the driving process defined through the It\^o differential equation
\begin{equation}
\de X_t = \sqrt{\kappa} \de B_t + \frac{\rho \de t}{X_t - Y_t} ,
\end{equation}
where $Y_t = g_t(y)$. For details of this coordinate change and time change see \cite{schramm-wilson-2005-}.

The process $(X_t-Y_t)/\sqrt{\kappa}$ is, in fact, a Bessel process.
The parameter $\rho$ depends on $\theta$ and $\kappa$ through
\begin{equation} \label{eq: rho theta}
\rho = \pm \theta + \frac{\kappa - 6}{2} ,
\end{equation}
where the sign depends on which of the points $-\infty$ or $+\infty$ was mapped to $\infty$. The law of $K_t$ of the above driving process
is called SLE$_\kappa(\rho)$.

This description works until the stopping time
\begin{equation} \label{eq: disconnect time}
\tau = \inf \{ t \geq 0 : |Y_s - X_s| \to 0 \textrm{ as } s \nearrow t \} .
\end{equation}
For the strip SLE this is the time when the curve disconnects $-\infty$ from $+\infty$ that is the curve hits $i\pi + \R$.
After this the strip SLE can't be continued in any straightforward way.
For the upper half-plane SLE $\tau$ is the time when the curve disconnects $y$ from $\infty$ (for $\kappa >4$)
or the curve hits $y$ (for $\kappa \leq 4$). 
After time $\tau$ the upper half-plane SLE can be continued, at least for a range of values of the parameters.

SLE$_\kappa (\rho)$ are important
since they are the random curves of the upper half-plane that depend on three marked points and satisfy Schramm's
principle. 
And especially important is the case $\rho = \kappa - 6$ since that is the coordinate transformation
of chordal SLE under a M\"obius map taking the points $0$ and $\infty$ to two points $x$ and $y$ on the real line.
This can be seen from the equation~\eqref{eq: rho theta}: 
since $\rho = 0$ is the chordal SLE, $\rho = \kappa -6$ must be the coordinate change of chordal SLE. 

Since for $\kappa \in (0,8)$ the chordal SLE avoids almost surely a given point in $\overline{\half} \setminus \{0\}$, 
it avoids especially the point that is mapped to $\infty$. 
From this it follows that the image of the full trace $\gamma (0,\infty)$ under the M\"obius map
is a bounded set. Hence considering SLE$_\kappa(\kappa - 6)$ makes it possible to study the properties of the full trace
of chordal SLE$_\kappa$.

It is also easy to see from the equation~\eqref{eq: rho theta} that if the interface of an Ising type model with
$(+, -, \textrm{free})$-boundary condition has a scaling limit that is SLE$_\kappa(\rho)$
then it has to be $\theta=0$ and $\rho = (\kappa - 6)/2$. This special case is also called dipolar SLE,
see \cite{bauer-bernard-houdayer-2005-}.

\section{Stationarity and some example calculations} \label{sec: moments}

\subsection{Stationarity of SLE} \label{ssec: stationarity}

Now we are ready to present the key idea of this paper. We will take a random conformal mapping and require that its
law is invariant under SLE flow. Such a random conformal mapping is said to have \emph{stationary} law. Based on this
invariance we can derive equations satisfied by quantities related to SLE.

Let $x,y \in \R$, $x \neq y$. Consider SLE$_\kappa (\rho)$ so that $X_0 = x$ and $Y_0 = y$, $X_t$ is the driving process,
$Y_t$ is as above, and $g_t$ is the Loewner map. 
Let $\phi_t(z)=\alpha_t z + \beta_t$ be the transformation that maps the points $x$ and $y$ to the points $X_t$ and $Y_t$.
We require that
\begin{equation}
 \left\{
 \begin{aligned}
   \phi_t(x) & = X_t \\ 
   \phi_t(y) & = Y_t .
 \end{aligned}
 \right.
\end{equation}
From these equations we solve the processes $\alpha_t$ and $\beta_t$.

Consider a random conformal map $\tilde{G}:\half \setminus \tilde{K} \to \half$ that is normalized 
by $\tilde{G} (z) = z + \oo(1)$ at the
infinity, and independent from the SLE given by $X_t$ and preserved by the SLE flow in the following sense: the mapping
\begin{equation}
G_t = \phi_t \circ \tilde{G} \circ \phi_t^{-1} \circ g_t
\end{equation}
has the same law as $\tilde{G}$. This property is schematically illustrated in Figure~\ref{fig: stationarity}. 
The following theorem tells that the mapping 
$\tilde{G}$ should be thought as $\tilde{g}_{\tilde{\tau}}$ where $\tilde{g}_t$
is SLE$_\kappa (\rho)$ and independent of $g_t$, and $\tilde{\tau}$ is the stopping time
defined analogously as in the equation~\eqref{eq: disconnect time}.

\begin{figure}[htb]
\begin{center}
\includegraphics[width=.975\textwidth]{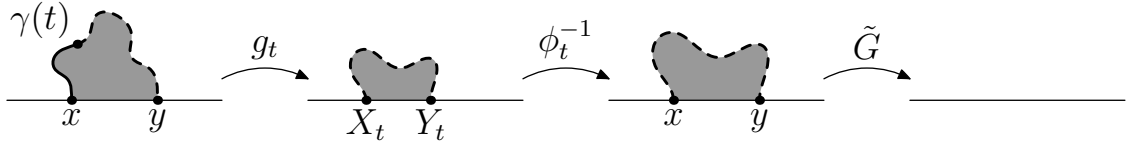}
\caption{The law of $\tilde{G}$ is stationary in the following sense: if the law of the hull in the third picture is
taken according to the law of $\tilde{G}$ and if an independent piece of SLE is added as in the first picture, then
the law of this modified hull is the same as the first one.}
\label{fig: stationarity}
\end{center}
\end{figure}

\begin{theorem} \label{thm: stationarity}
Let the pair $(g_t,\tau)$ be SLE$_\kappa(\rho)$ and the stopping time of the equation~\eqref{eq: disconnect time},
and let $(\tilde{g},\tilde{\tau})$ be an independent copy of them. 
If $\rho < (\kappa - 4)/2$ then $\tau < \infty$ a.s. and hence $g_\tau$ is well-defined.
Furthermore, if $\phi_t$ is as above,
then
$\tilde{G} = \tilde{g}_{\tilde{\tau}}$ and
\begin{equation} \label{eq: stationarity}
G_t = \begin{cases}
 \phi_t \circ \tilde{G} \circ \phi_t^{-1} \circ g_t & \textrm{ on } \{\tau > t\} \\
    g_\tau                                          & \textrm{ on } \{\tau \leq t\}
 \end{cases}
\end{equation}
are identically distributed.
\end{theorem}

\begin{proof}
The argument we present here is basically that SLE$_\kappa(\rho)$ satisfies Schramm's principle for three marked points.
Since we didn't provide the details above, it is worth writing down.

Assume that $x<y$. The other case can be done symmetrically.
Write the Bessel stochastic differential equation in a bit non-standard way as
\begin{equation} \label{eq: bessel sde}
\de Z_t = \sqrt{\kappa} \de B_t + (\rho + 2) \frac{\de t}{Z_t} .
\end{equation}
Let $Z_t$ and $\tilde{Z}_t$ be the solutions of \eqref{eq: bessel sde}
for two independent Brownian motions and with the initial condition $Z_0 = \tilde{Z}_0 = y-x$. 
Now the driving process $X_t$ is defined through the equations
\begin{align*}
Y_t & = Y_0 + \int_0^t \frac{2\de s}{Z_s} \\
X_t & = Y_t - Z_t .
\end{align*}
In the same way using $\tilde{Z}_t$ instead of $Z_t$ define $\tilde{X}_t$ and $\tilde{Y}_t$.
The stopping time $\tau$ can be written as
\begin{equation*}
\tau = \inf \{ t \geq 0 : Z_s \to 0 \textrm{ as } s \nearrow t \}
\end{equation*}
and $\tilde{\tau}$ can be written using $\tilde{Z}_t$.

The first claim follows from the fact that $Z_t$ is a scaled version of a Bessel process defined using the standard normalization, 
with the index 
\begin{equation*}
\nu = 2 \frac{\rho + 2}{\kappa} .
\end{equation*}
A standard fact is that a Bessel process will hit $0$ if and only if $\nu < 1$, see Example~6.5.3 of \cite{durrett-1996b-}.

The mapping $\phi_t \circ \tilde{g}_s \circ \phi_t^{-1}$ satisfies the normalization
\begin{equation*}
\phi_t \circ \tilde{g}_s \circ \phi_t^{-1} (z) = z + \frac{2 \alpha_t^2 s }{z} +\ldots
\end{equation*}
and the family of mappings $\hat{g}_s = \phi_t \circ \tilde{g}_{s/\alpha_t^2} \circ \phi_t^{-1}$ satisfies the Loewner
equation with the driving process 
\begin{align*}
\hat{X}_s & =\alpha_t \tilde{X}_{s/\alpha_t^2} + \beta_t
            = \alpha_t \left( \tilde{Y}_0 - \tilde{Z}_{s/\alpha_t^2} 
            + \int_0^{s/\alpha_t^2} \frac{2\de u}{\tilde{Z}_u} \right) + \beta_t \\
          & = Y_t - \alpha_t \tilde{Z}_{s/\alpha_t^2} 
            + \alpha_t \int_0^{s/\alpha_t^2} \frac{2\de u}{\tilde{Z}_u} .
\end{align*}
Since the second and third term satisfy the Brownian scaling we can write
\begin{align*}
\hat{X}_s  
 & = Y_t -  \hat{Z}_{s} + \int_0^{s} \frac{2\de u}{\hat{Z}_u} 
\end{align*}
where $\hat{Z}_{s}$ is a solution of the Bessel SDE~\eqref{eq: bessel sde} with the initial value  
$ \hat{Z}_{0} = \alpha_t Z_0 = Y_t - X_t$. 
Hence the process defined as
\begin{equation*}
\begin{cases}
Z_s & \textrm{when } s \leq t \\
\hat{Z}_{s-t} & \textrm{when } s > t
\end{cases}
\end{equation*}
is distributed as the process $Z_s$ and $\hat{g}_s \circ g_t$ is distributed as $g_{t+s}$. Let
$\sigma$ be the stopping time for $\hat{Z}_s$ hitting $0$ as $s \nearrow \sigma$. Then $\sigma = \tilde{\tau}\alpha_t^2$.
And hence on $\{\tau > t\}$ the mapping $\phi_t \circ \tilde{G} \circ \phi_t^{-1} \circ g_t$ has the same law
as $g_\tau$. On $\{\tau \leq t\}$ the statement follows immediately. 
\end{proof}

For small $t$, the event $\{\tau \leq t\}$ has exponentially small probability. To see this we need to consider
only the diffusion term ($\de B_t$) of the equation~\eqref{eq: bessel sde} and we need to note that the probability that 
a Brownian motion started from $y-x$ comes near $0$ in the time interval $[0,t]$ is exponentially small in $1/t$.
By this property we need basically just care about the first case of the equation~\eqref{eq: stationarity}.
Actually we will use the stationarity to calculate the distribution of $\tau$. See the equation~\eqref{eq: a1 distribution}
below.

Write the expansion of $g_t$ as
\begin{equation}
g_t(z) = z + \frac{a_1(t)}{z} + \frac{a_2(t)}{z^2} + \ldots
\end{equation}
We call \emph{SLE data} the collection of random variables
\begin{equation}
X_t, Y_t, a_1(t), a_2(t), \ldots
\end{equation}
SLE data carries all the information about $g_t$ and the law of $g_s$, $s>t$. The coefficient
$a_1(t)=2t$ and the higher coefficient are definite integrals of polynomials on the lower coefficients and $X_t$.
So in principle, they could be calculated. 
On the stopping time $\tau$ we have $X_t - Y_t \to 0$ as $t \nearrow \tau$ and then the SLE data simplifies
to $a_k(\tau)$, $k \in \N$. Note that also $a_1(\tau)=2\tau$ is random.

During the rest of this paper we will present some examples how to use the stationarity to calculate SLE data related quantities,
like the moments $\E[ \prod a_{k_i}(\tau) ]$.

It should be stressed, that the expected value $\E[ \prod a_{k_i}(\tau) ]$ exists only for a certain range of the parameters $\kappa,\rho$.
For example, when $\rho = \kappa -6$, for any $\kappa < 8$, $\tau < \infty$ a.s. and $g_\tau$ is well-defined, 
but $\E[ \prod |a_{k_i}(\tau)| ] < \infty$ only for $0 \leq \kappa < \kappa_0(k_1,\ldots,k_n)$ where $\kappa_0(k_1,\ldots,k_n) \to 0$ as
a natural degree of $(k_1,\ldots,k_n)$ grows.
This will be commented more in the end of Section~\ref{ssec: general moments}.

\subsection{Basic equations for the coefficients of $\tilde{G}$}

In this section we derive the equation describing the flow of $(\tilde{a}_k)$ under
the flow \eqref{eq: stationarity}.
Use the expansion
\begin{displaymath}
\tilde{G}(z)= z + \frac{\tilde{a}_1}{z} + \frac{\tilde{a}_2}{z^{2}} + \ldots
\end{displaymath}
to write the expansion of $G_t$ of the equation~\eqref{eq: stationarity}
\begin{align}
G_t(z) &= \alpha_t \tilde{G} \left( \frac{g_t(z) - \beta_t}{\alpha_t} \right) + \beta_t \nonumber \\
       &= g_t(z) + \frac{\tilde{a}_1 \alpha_t^{2}}{g_t(z) - \beta_t} + \frac{\tilde{a}_2 \alpha_t^{3}}{(g_t(z) - \beta_t)^{2}}
            + \ldots
\label{Gexpa}
\end{align}
So to get the It\^o differential of the expansion
we need to calculate It\^o differential of $g_t(z)$ and expressions of type $\alpha_t^{n+1}/(g_t(z) - \beta_t)^{n}$ at time $t=0$. 

Let's simplify the setup: let $\sigma \in \{-1,1\}$ and $x=\sigma$ and $y=-\sigma$. Note we can always transform the above setup
to this simplified setup with scaling and translation. Now
\begin{equation}
\de g_t(z) \underset{t=0}{=} \frac{2 \de t}{z - \sigma} = \left\{\frac{2}{z} + \sigma \frac{2}{z^2} + \frac{2}{z^3} 
+ \sigma \frac{2}{z^4} + \ldots \right\} \de t
\end{equation}
and after a short calculation we find that
\begin{align}
\de & \frac{ \alpha_t^{n+1} }{ (g_t(z) - \beta_t)^{n} } \underset{t=0}{=}  \left\{ \left[ (n+1) \frac{\rho+2}{4}
+ n(n+1)\frac{\kappa}{8} \right] \frac{1}{z^n} \right. \nonumber\\
& + \sigma \left[ n \frac{\rho-2}{4} + n(n+1)\frac{\kappa}{4} \right] \frac{1}{z^{n+1}}
+   \left[ -2 n + n(n+1)\frac{\kappa}{8} \right] \frac{1}{z^{n+2}} \nonumber \\
& - \sigma \left. \frac{2 n}{z^{n+3}} - \frac{2 n}{z^{n+4}} - \ldots \right\} \de t + \left\{\sigma (n+1)\frac{1}{z^{n}}
+  n\frac{1}{z^{n+1}} \right\} \de B_t .
\end{align}
Using the notation $G_t(z)= z + a_1(t) z^{-1} + a_2(t) z^{-2} + \ldots$ and combining last two It\^o differentials with (\ref{Gexpa})
we finally get
\begin{align}
\de a_n(t) \underset{t=0}{=} & \bigg\{  \frac{1}{8} (n+1)(\kappa n + 2\rho +4) \tilde{a}_n
 + \sigma \frac{1}{4} (n-1)(\kappa n + \rho -2) \tilde{a}_{n-1} \nonumber \\
& +  \frac{1}{8} (n-2)(\kappa (n-1) -16) \tilde{a}_{n-2} - \sum_{k=1}^{n-3} 2 k \sigma^{n-k} \tilde{a}_{k}
+ 2 \sigma^{n+1} \bigg\} \de t \nonumber \\
& + \frac{\sqrt{\kappa}}{2} \bigg\{ \sigma(n+1) \tilde{a}_n + (n-1) \tilde{a}_{n-1} \bigg\} \de B_t .
\end{align}
From now on we will not distinguish between $\tilde{a}_n$ and $a_n$. Write in short
\begin{equation}
\de a_n = \left(c_{n,0} + \sum_{k=1}^n c_{n,k} a_k \right) \de t + \left(d_{n,n-1} a_{n-1} + d_{n,n} a_n 
\right) \de B_t . \label{eq: cdnot}
\end{equation}
These expressions are linear in variables $(a_k)$ and hierarchical in the sense that the It\^o differential of 
$a_n$ involves only terms $a_k$ for $k \leq n$. This is really the reason why this method is useful.

\subsection{Stationarity for the inverse mapping}

Similar argument can be made for the inverse mapping $\tilde{F}=\tilde{G}^{-1}:\half \to \half \setminus \tilde{K}$.
For the inverse mapping $f_t$ of $g_t$ the Loewner equation is
\begin{equation}
\partial_t f_t(z)=- f_t'(z) \frac{2}{z - X_t} .
\end{equation}

Let $\tilde{F}$ be a random conformal mapping that is preserved by SLE flow of $f_t$ in the following sense:
the mapping
\begin{equation}
F_t = f_t \circ \phi_t \circ \tilde{F} \circ \phi_t^{-1}
\end{equation}
has the same law as $\tilde{F}$.

Now $F_t(z)=f_t(\alpha_t\tilde{F}((z-\beta_t)/\alpha_t)+\beta_t)$ and therefore
\begin{equation}
\de F_t(z) \underset{t=0}{=} -\frac{2 \de t}{\tilde{F}(z) - \sigma} + \de \left( 
\alpha_t \tilde{F}\left( \frac{z-\beta_t}{\alpha_t} \right) + \beta_t \right) .
\end{equation}
If $\tilde{F}(z) = z + \tilde{b}_1 z^{-1} + \tilde{b}_2 z^{-2} + \ldots$ and $F_t(z) = z + b_1(t) z^{-1} 
+ b_2(t) z^{-2} + \ldots$, we get expression for $\de b_n(t=0)$ in terms of $\tilde{b}_m$ similarly as in the case of
$\tilde{G}$. But now the expressions are not linear in $\tilde{b}_m$. For this reason we won't consider this setup.

\subsection{The reversibility of SLE with moments} \label{ssec: reversibility}

The reversibility of SLE is the following property: let $\gamma$ be chordal SLE from $0$ to $\infty$. Then $\gamma$ and
$-1/\gamma$ appropriately parameterized have the same law. In terms of SLE$_\kappa(\kappa - 6)$ this can be stated as 
SLE$_\kappa(\kappa - 6)$ from $x$ to $y$ and 
SLE$_\kappa(\kappa - 6)$ from $y$ to $x$ appropriately parameterized have the same law.
Especially this means that the hulls of the full traces have to have the same law.

Consider now $x = -1$ and $y = 1$. Start SLE$_\kappa(\kappa -6)$ from $x$ and denote by $\tau_{-}$ the hitting time
of $y$ and let the conformal map be $g_{\tau_{-}}^{-}(z)= z + a_1^{-} z^{-1} + a_2^{-} z^{-2} + \ldots $.
In the same way start SLE$_\kappa(\kappa -6)$ from $y$ and denote by $\tau_{+}$ the hitting time
of $x$ and let the conformal map be $g_{\tau_{+}}^{+}(z)= z + a_1^{+} z^{-1} + a_2^{+} z^{-2} + \ldots $.
The reversibility can be formulated using the coefficient $a_n^{\pm}$: for any $n \in \N$ and 
$l_1,\ldots,l_n \in \N$, $l_1 < l_2 < \ldots < l_n$
\begin{displaymath}
(a^{-}_{l_1},a^{-}_{l_2},\ldots,a^{-}_{l_n}) \overset{\mathcal{L}}{=} (a^{+}_{l_1},a^{+}_{l_2},\ldots,a^{+}_{l_n}) .
\end{displaymath}
i.e. they have the same law.

Let $m(z)=-\overline{z}$. This map is the mirror map that changes $x$ with $y$ and therefore $m \circ g_{\tau_{-}}^{-} \circ
m \overset{\mathcal{L}}{=} g_{\tau_{+}}^{+}$. On the other hand for any $g(z) = z + a_1 z^{-1} + a_2 z^{-2} + \ldots$
with real $a_m$, $m \in \N$, we have
\begin{align*}
m \circ g \circ m (z) & = m(g(-\overline{z})) = m\left(-\overline{z} - \frac{a_1}{\overline{z}} + \frac{a_2}{\overline{z}^2}
 - \frac{a_3}{\overline{z}^3} + \ldots \right) \\
& = z + \frac{a_1}{z} - \frac{a_2}{z^2} + \frac{a_3}{z^3} + \ldots 
\end{align*}
In words, the even coefficients change sign under the mirror map $m$. This shows that the reversibility is equivalent
to
\begin{displaymath}
(a^{-}_{l_1},a^{-}_{l_2},\ldots,a^{-}_{l_n}) \overset{\mathcal{L}}{=} ((-1)^{l_1+1}a^{-}_{l_1},
(-1)^{l_2+1}a^{-}_{l_2},\ldots,(-1)^{l_n+1}a^{-}_{l_n}) 
\end{displaymath}
which a nice way to give a concrete formulation for the reversibility.

Let $n \in \N$ and $(k_1,\ldots,k_n) \in \{0,1,2,\ldots\}^n$. If the reversibility holds then
\begin{equation} \label{eq: moments reversibility}
\E \left[ \prod_{j=1}^n (a_j^{-})^{k_j} \right] = (-1)^{\sum_{1\leq i \leq n/2} k_{2 i} } 
\E \left[ \prod_{j=1}^n (a_j^{-})^{k_j} \right]
\end{equation}
which should vanish when $\sum_{1\leq i \leq n/2} k_{2 i}$ is odd. In fact, if every moment existed, 
one strategy in proving the reversibility, at least in the case $\kappa \in (0,4]$, could be
showing that these odd moments vanish and showing that the moments determine the distribution.

\subsection{General expression for moments} \label{ssec: general moments}

To work out equations for expected values of the type in the equation~\eqref{eq: moments reversibility} we use the following notation:
fix $n \in \N$ and $(k_1,\ldots,k_n) \in \{0,1,2,\ldots\}^n$ and let
\begin{displaymath}
\Pi = \Pi(k_1,k_2,\ldots,k_n)=a_1^{k_1} a_2^{k_2} \cdot \ldots \cdot a_n^{k_n}
\end{displaymath}
and for $i \in \{1,\ldots, n\}$
\begin{align*}
\Pi^i (k_1,k_2,\ldots,k_n) & = \Pi(k_1,\ldots, k_{i-1}, k_i+1, k_{i+1},\ldots,k_n) \\
\Pi_i (k_1,k_2,\ldots,k_n) & = \Pi(k_1,\ldots, k_{i-1}, k_i-1, k_{i+1},\ldots,k_n) .
\end{align*}
Here $\Pi=0$ with negative arguments. Define similarly $\Pi_{i_1,\ldots,i_l}^{j_1,\ldots,j_m}$. Further $\Pi^0 = \Pi$.
Since we are looking for the stationary $\tilde{G}$ we require that the expectation of the drift of $\Pi$ vanishes.
So for a while we will manipulate the expression of $\de \Pi$.

Using this notation and the notation of equation (\ref{eq: cdnot}) we find that
\begin{align}
\de \Pi & = \sum_i \Pi_i \de a_i + \frac{1}{2} \sum_{i,j} k_i ( k_j - \delta_{ij} ) \Pi_{i,j}
   \de a_i \de a_j \nonumber \\
 & = \ldots = \Bigg\{ \Bigg[ \frac{1}{2} \sum_i k_i d_{ii} \Big( 2 \frac{c_{ii}}{d_{ii}} - d_{ii} + \sum_j k_j d_{jj}
   \Big) \Bigg] \Pi \nonumber \\
 & + \sum_{i>1} k_i d_{i,i-1} \Big( \frac{c_{i,i-1}}{d_{i,i-1}} - d_{ii} + \sum_j k_j d_{jj} \Big) 
   \Pi_i^{i-1} \nonumber \\
 & + \frac{1}{2} \sum_{i,j>1} k_i (k_j -\delta_{ij}) d_{i,i-1} d_{j,j-1} \Pi_{i,j}^{i-1,j-1} \nonumber \\
 & + \sum_{i} k_i \sum_{l=0}^{i-2} c_{i l} \Pi_{i}^{l} \Bigg\} \de t + \big\{ \quad \big\} \de B_t .
\end{align}
Note that the following expressions are independent of the summation index $i$ for any $\kappa$ and $\rho$
\begin{align*}
2 \frac{c_{ii}}{d_{ii}} - d_{ii} & = \frac{2 \rho + 4 - \kappa}{2 \sqrt{\kappa}} \\
\frac{c_{i,i-1}}{d_{i,i-1}} - d_{ii} & = \frac{\rho -2 - \kappa}{2 \sqrt{\kappa}} .
\end{align*}
Next we write that $\sum_i k_i d_{ii} = \sqrt{\kappa} /2 \sum_i k_i (i+1) = \sqrt{\kappa} N$, which defines the
degree 
\begin{equation}
N = \frac{1}{2} \sum_i k_i (i+1)
\end{equation}
of a moment $\Pi$.
Plugging this and the values of $c_{ij}$ and $d_{ij}$ we get that
\begin{align}
\de \Pi & = \Bigg\{ \frac{1}{4} N [ (2 \rho + 4 - \kappa) + 2 N \kappa ] \Pi  
 + \sigma \frac{1}{4} [(\rho-2-\kappa) + 2N\kappa] \sum_{i} k_i (i-1) \Pi_i^{i-1} \nonumber \\
 & + \frac{\kappa}{8} \sum_{i,j} k_i (k_j -\delta_{ij}) (i-1)(j-1) \Pi_{i,j}^{i-1,j-1} \nonumber \\
 & - 2 \sum_{i} k_i \sum_{l=1}^{i-2} \sigma^{i-l} l \; \Pi_{i}^{l} 
   + 2 \sum_{i} k_i \sigma^{i+1} \Pi_{i}^{0}
\Bigg\} \de t + \big\{ \quad \big\} \de B_t . \label{eq: dePikr}
\end{align}
For $\rho = \kappa -6$ the above brackets are $2 \rho + 4 - \kappa + 2 N \kappa = (2N+1)\kappa -8$ and
$\rho - 2 - \kappa + 2N\kappa = 2N\kappa - 8$. Let's use this value of $\rho$ for a while.

Now we analyze the degree N. First of all
\begin{align*}
N & = \frac{1}{2} \sum_i k_{2 i -1} \cdot 2i + \frac{1}{2} \sum_i k_{2 i} \cdot (2i+1) \\
 & = \sum_i (k_{2 i -1} + k_{2i} )i + \frac{1}{2} \sum_i k_{2 i} .
\end{align*}
So $N$ is either a half-integer or an integer depending whether $\sum_i k_{2i}$ is odd or even. So for the reversibility 
we would like to show that $\E [\Pi ]=0$ when $N$ is a half-integer. Next we
note that the drift in the equation (\ref{eq: dePikr}) decomposes into $A + \sigma B$ where $A$ and $B$
don't depend (directly) on $\sigma$ and all the half-integer moments are put in the other one and the
integer moments on the other.

Under the reversibility $\E [\Pi]=0$ when $N$ is a half-integer, then for $N$ an integer we would have
\begin{equation} \label{eq: moments general form}
\E [ \Pi(k_1,\ldots,k_n) ] = \frac{ p_{k_1,\ldots,k_n} (\kappa) }{%
(8 - 3 \kappa)(8 - 5 \kappa)\cdot \ldots \cdot (8-(2N+1)\kappa)} ,
\end{equation}
where $p_{k_1,\ldots,k_n}$ is a polynomial with highest degree $\tilde{N} = 1/2 \sum_i k_i (i-1)$. Denominator
follows from the fact that as we recursively solve $\E[\Pi]$ from (\ref{eq: dePikr}) by demanding that the
drift vanishes, the factor in front of the moment with the largest degree is $1/4 N [(2N+1)\kappa -8]$. Similarly
$\kappa$ can enter numerator only through the term $\Pi_{i,j}^{i-1,j-1}$ (this argument requires more care though).
$\tilde{N}$ is the number of steps from $\Pi(k_1,\ldots,k_n)$ to $\Pi(k_1',0,0,\ldots,0)$ by lowering two powers
with $\Pi_{i,j}^{i-1,j-1}$.

The equation~\eqref{eq: moments general form} can be interpreted so that the expected value $\Pi(k_1,\ldots,k_n)$
exists for small $\kappa$ as long as the right-hand side is finite. So we can read from this general form
that the expected value $\Pi(k_1,\ldots,k_n)$ exists for $\kappa \in \big(0,8/(2N+1)\big)$.
This result is proven in Appendix~A.1 of \cite{kytola-2006-}. The result therein includes both cases the half-integer
and the integer moments.

\subsection{Calculating moments $a_1^n$, $a_1^n a_2^m$ and so on} \label{ssec: a1 a2 moments}

In this section, we study only the case $\rho = \kappa -6$. We will show how to actually calculate moments, i.e. expected
values of SLE data.
Let's calculate It\^o differential
\begin{align*}
\de(a_1^n) &= n a_1^{n-1} \de a_1 + \frac{1}{2} n(n-1) a_1^{n-2} (\de a_1)^2 \\
 &= \left[ n a_1^{n-1} \left(2 + \frac{3\kappa-8}{4} a_1\right) + \frac{1}{2} n(n-1) a_1^{n-2} \cdot \kappa a_1^2
    \right] \de t + ( \quad ) \de B_t \\
 &= n \left[ 2 a_1^{n-1} + \frac{(2n+1)\kappa-8}{4} a_1^{n} \right] \de t + ( \quad ) \de B_t .
\end{align*}
Then we demand that expectation of the drift is zero. This gives
\begin{displaymath}
\E [a_1^n] = \frac{8 \E[a_1^{n-1}]}{8-(2n+1)\kappa} = \ldots 
= \frac{8^{n}}{(8-3\kappa)(8-5\kappa)\cdot\ldots\cdot(8-(2n+1)\kappa)}
\end{displaymath}
since $\E[a_1^{0}]=1$. This is true for $x=\sigma$ and $y=-\sigma$. For general $x,y \in \R$, use a suitable M\"obius
transformation to get
\begin{equation} \label{eq: a1 formula}
\E [a_1^n] 
= \frac{2^{n} (x-y)^{2n}}{(8-3\kappa)(8-5\kappa)\cdot\ldots\cdot(8-(2n+1)\kappa)} .
\end{equation}

Similar calculation for $a_1^{n} a_2^{m}$, $m$ even, gives
\begin{equation}
\E [a_1^n a_2^m] \label{eq: a1a2 formula}
= \frac{2^{2n+3m} \left(\frac{\kappa}{6}\right)^{m/2} \frac{m!}{(\frac{m}{2})!}}{
(8-3\kappa)(8-5\kappa)\cdot\ldots\cdot(8-(2n+3m+1)\kappa)} .
\end{equation}
The higher moments can be in principle calculated using the recursion we get from the equation~\eqref{eq: dePikr}.
The author hasn't been able to completely solve the recursion.

\subsection{Density function of $a_1$} \label{ssec: a1 density}

As stated earlier $a_1$ is distributed as $2 \tau$ where $\tau$ is the hitting time of $0$ for a Bessel process.
Its distribution could be calculated using a martingale trick or similarly as below but using just the Bessel process. 
However the following way to calculate the distribution is worth mentioning.

If the capacity $a_1(t)$ has a density function $\nu_t$ then
\begin{displaymath}
\E_t [ f(a_1) ] = \int_0^\infty f(x) \nu_t (x) \de x
\end{displaymath}
for each sufficiently smooth $f:(0,\infty) \to \R$ with compact support. For such function It\^o differential
is
\begin{displaymath}
\de f(a_1) = \left[ \left(2 + \frac{\kappa+2\rho+4}{4} a_1 \right) f'(a_1) + \frac{\kappa}{2} a_1^2 f''(a_1) \right] \de t
 + \sigma a_1 \sqrt{\kappa} f'(a_1) \de B_t .
\end{displaymath}
For $\nu_t = \nu$ stationary, the expectation of the drift has to vanish
\begin{align}
0 & = \E \left[ \left(2 + \frac{\kappa+2\rho+4}{4} a_1 \right) f'(a_1) + \frac{\kappa}{2} a_1^2 f''(a_1) \right] \nonumber \\
  & = \int_0^\infty [p(x) f''(x) + q(x) f'(x) ]\nu (x) \de x 
  \label{eq: cap_df}
\end{align}
where $p(x)=\kappa/2 x^2$ and $q(x)=2 + (\kappa+2\rho+4)/4 x$.
Since equation (\ref{eq: cap_df}) holds for every $f$ smooth and with compact support, we conclude
$-(p(x) \nu(x))' + q(x) \nu (x) = C = \textrm{ const.}$ If we assume $\nu$ and $\nu'$ go zero as $x \to 0$, then $C=0$.

Now we solve
\begin{displaymath}
\frac{\nu'(x)}{\nu(x)} = \frac{q(x)-p'(x)}{p(x)} = -\frac{3 \kappa - 2 \rho - 4}{2 \kappa} \frac{1}{x}
+ \frac{4}{\kappa}\frac{1}{x^2}
\end{displaymath}
giving
\begin{equation} \label{eq: a1 distribution}
\nu(x) = C_{\kappa,\rho} x^{-\frac{3 \kappa - 2 \rho - 4}{2 \kappa}} e^{-\frac{4}{\kappa}\frac{1}{x}} .
\end{equation}
Coefficient $C_{\kappa,\rho}$ is determined from $\int_0^\infty \nu (x) \de x=1$, where the integral converges
if and only if the power of $x$ is smaller than $-1$. For $\rho = \kappa -6$ this means $\kappa < 8$. 
This result can be explained as follows: for $\kappa<8$ the chordal SLE a.s. avoids given point and hence the capacity
seen from this point is a.s. finite.

\section{Conclusions}

It was shown how to formulate the stationarity of SLE$_\kappa(\rho)$ as stationarity of the law of a stopped hull
under a SLE induced flow. 
One of the advances of this approach is that it involves the full SLE trace directly. The full trace is the most interesting object
from the statistical physics point of view.

When using the approach to calculate the moments $\E[ \prod a_{k_j} ]$, the problem is that these expected values only exist
for a range of the parameter $\kappa$. Hence the approach should be applied in some different way. 
For example, some other function of the random variables $a_1,a_2,\ldots$ could be taken, say, 
such as $\E[ e^{i \lambda a_1 } \prod a_{k_j} ]$.
As proposed by Stanislav Smirnov,
one option is to try to find an alternative interpretation beyond the blowup for the analytic continuations of 
the moment formulas such as \eqref{eq: a1 formula} and \eqref{eq: a1a2 formula}.

\section*{Acknowledgments}

I wish to thank Stanislav Smirnov. This work was started during a joint project. I wish also thank Kalle Kyt\"ol\"a and
Antti Kupiainen for useful discussions. Jan Cristina also deserves thanks for reading a part of the paper
and for the discussions on writing in English.
This work was financially supported by Academy of Finland and
by Finnish Academy of Science and Letters, Vilho, Yrj\"o and Kalle V\"ais\"al\"a Foundation.

\end{document}